\documentclass[twoside]{article}%
\usepackage{amssymb}
\usepackage{amsfonts}
\usepackage{amsmath}
\usepackage{graphicx}%
\providecommand{\U}[1]{\protect\rule{.1in}{.1in}}
\newtheorem{theorem}{Theorem}

\newtheorem{corollary}[theorem]{Corollary}

\newtheorem{definition}[theorem]{Definition}

\newtheorem{lemma}[theorem]{Lemma}

\newtheorem{remark}[theorem]{Remark}

\newenvironment{proof}[1][Proof]{\noindent\textbf{#1.} }{\ \rule{0.5em}{0.5em}}
\begin{document}

\title{\textbf{Analytical} \textbf{Blowup Solutions to the }$3$\textbf{-dimensional
Pressureless Navier-Stokes-Poisson Equations with Density-dependent Viscosity}}
\author{Y\textsc{uen} M\textsc{anwai\thanks{E-mail address: nevetsyuen@hotmail.com }}\\\textit{Department of Applied Mathematics, }\\\textit{The Hong Kong Polytechnic University,}\\\textit{Hung Hom, Kowloon, Hong Kong}}
\date{Revised 29-Oct-2008}
\maketitle

\begin{abstract}
We study the pressureless Navier--Stokes-Poisson equations of describing the
evolution of the gaseous star in astrophysics. The isothermal blowup solutions
of Yuen, to the Euler-Poisson equations in $R^{2}$, can be extended to the
pressureless Navier-Stokes-Poisson equations with density-dependent viscosity
in $R^{3}$. Besides some remarks, about the meaning of the blowup solutions
and the applicability of such solutions to the the drift-diffusion model in
semiconductors, are discussed in the end.

\end{abstract}

\section{Introduction}

The evolution of a self-gravitating fluid such as gaseous stars can be
formulated by the Navier-Stokes-Poisson equation of the following form:
\begin{equation}
\left\{
\begin{array}
[c]{rl}%
{\normalsize \rho}_{t}{\normalsize +\nabla\bullet(\rho u)} & {\normalsize =}%
{\normalsize 0,}\\
{\normalsize (\rho u)}_{t}{\normalsize +\nabla\bullet(\rho u\otimes u)+\nabla
P} & {\normalsize =}{\normalsize -\rho\nabla\Phi+vis(\rho,u),}\\
{\normalsize \Delta\Phi(t,x)} & {\normalsize =\alpha(N)}{\normalsize \rho,}%
\end{array}
\right.  \label{Euler-Poisson}%
\end{equation}
where $\alpha(N)$ is a constant related to the unit ball in $R^{N}$:
$\alpha(1)=2$; $\alpha(2)=2\pi$ and For $N\geq3,$%
\[
\alpha(N)=N(N-2)V(N)=N(N-2)\frac{\pi^{N/2}}{\Gamma(N/2+1)},
\]
where $V(N)$ is the volume of the unit ball in $R^{N}$ and $\Gamma$ is a Gamma
function. And as usual, $\rho=\rho(t,x)$ and $u=u(t,x)\in\mathbf{R}^{N}$ are
the density, the velocity respectively. $P=P(\rho)$\ is the pressure.

In the above system, the self-gravitational potential field $\Phi=\Phi
(t,x)$\ is determined by the density $\rho$ through the Poisson equation.

And $vis(\rho,u)$ is the viscosity function:%
\[
vis(\rho,u)=\bigtriangledown(\mu(\rho)\bigtriangledown\bullet u).
\]
Here we under a common assumption for:
\[
\mu(\rho)\doteq\kappa\rho^{\theta}%
\]
and $\kappa$ and $\theta\geq0$ are the constants. In particular, when
$\theta=0$, it returns the expression for the $u$ dependent only viscosity
function:%
\[
vis(\rho,u)=\kappa\Delta u.
\]
The equations (\ref{Euler-Poisson})$_{1}$ and (\ref{Euler-Poisson})$_{2}$
$(vis(\rho,u)\neq0)$ are the compressible Navier-Stokes equations with forcing
term. The equation (\ref{Euler-Poisson})$_{3}$ is the Poisson equation through
which the gravitational potential is determined by the density distribution of
the density itself. Thus, we call the system (\ref{Euler-Poisson}) the
Navier--Stokes-Poisson equations.

Here, if the $vis(\rho,u)=0$, the system is called the Euler-Poisson
equations.\ In this case, the equations can be viewed as a prefect gas model.
For $N=3$, (\ref{Euler-Poisson}) is a classical (nonrelativistic) description
of a galaxy, in astrophysics. See \cite{C}, \cite{M1} for a detail about the system.

$P=P(\rho)$\ is the pressure. The $\gamma$-law can be applied on the pressure
$P(\rho)$, i.e.%
\begin{equation}
{\normalsize P}\left(  \rho\right)  {\normalsize =K\rho}^{\gamma}\doteq
\frac{{\normalsize \rho}^{\gamma}}{\gamma}, \label{gamma}%
\end{equation}
which is a commonly the hypothesis. The constant $\gamma=c_{P}/c_{v}\geq1$,
where $c_{P}$, $c_{v}$\ are the specific heats per unit mass under constant
pressure and constant volume respectively, is the ratio of the specific heats,
that is, the adiabatic exponent in (\ref{gamma}). In particular, the fluid is
called isothermal if $\gamma=1$. With $K=0$, we call the system is pressureless.

For the $3$-dimensional case, we are interested in the hydrostatic equilibrium
specified by $u=0,S=\ln K$. According to \cite{C}, the ratio between the core
density $\rho(0)$ and the mean density $\overset{\_}{\rho}$ for $6/5<\gamma
<2$\ is given by%
\[
\frac{\overset{\_}{\rho}}{\rho(0)}=\left(  \frac{-3}{z}\dot{y}\left(
z\right)  \right)  _{z=z_{0}}%
\]
where $y$\ is the solution of the Lane-Emden equation with $n=1/(\gamma-1)$,%
\[
\ddot{y}(z)+\dfrac{2}{z}\dot{y}(z)+y(z)^{n}=0,\text{ }y(0)=\alpha>0,\text{
}\dot{y}(0)=0,\text{ }n=\frac{1}{\gamma-1},
\]
and $z_{0}$\ is the first zero of $y(z_{0})=0$. We can solve the Lane-Emden
equation analytically for%
\[
y_{anal}(z)\doteq\left\{
\begin{array}
[c]{ll}%
1-\frac{1}{6}z^{2}, & n=0;\\
\dfrac{\sin z}{z}, & n=1;\\
\dfrac{1}{\sqrt{1+z^{2}/3}}, & n=5,
\end{array}
\right.
\]
and for the other values, only numerical values can be obtained. It can be
shown that for $n<5$, the radius of polytropic models is finite; for $n\geq5$,
the radius is infinite.

Gambin \cite{G} and Bezard \cite{B} obtained the existence results about the
explicitly stationary solution $\left(  u=0\right)  $ for $\gamma=6/5$ in
Euler-Poisson equations$:$%
\begin{equation}
\rho=\left(  \frac{3KA^{2}}{2\pi}\right)  ^{5/4}\left(  1+A^{2}r^{2}\right)
^{-5/2},\label{stationsoluionr=6/5}%
\end{equation}
where $A$ is constant.\newline The Poisson equation (\ref{Euler-Poisson}%
)$_{3}$ can be solved as%
\[
{\normalsize \Phi(t,x)=}\int_{R^{N}}G(x-y)\rho(t,y){\normalsize dy,}%
\]
where $G$ is the Green's function for the Poisson equation in the
$N$-dimensional spaces defined by
\[
G(x)\doteq\left\{
\begin{array}
[c]{ll}%
|x|, & N=1;\\
\log|x|, & N=2;\\
\dfrac{-1}{|x|^{N-2}}, & N\geq3.
\end{array}
\right.
\]
In the following, we always seek solutions in spherical symmetry. Thus, the
Poisson equation (\ref{Euler-Poisson})$_{3}$ is transformed to%
\[
{\normalsize r^{N-1}\Phi}_{rr}\left(  {\normalsize t,x}\right)  +\left(
N-1\right)  r^{N-2}\Phi_{r}{\normalsize =}\alpha\left(  N\right)
{\normalsize \rho r^{N-1},}%
\]%
\[
\Phi_{r}=\frac{\alpha\left(  N\right)  }{r^{N-1}}\int_{0}^{r}\rho
(t,s)s^{N-1}ds.
\]

\begin{definition}
[Blowup]We say a solution blows up if one of the following conditions is
satisfied:\newline(1)The solution becomes infinitely large at some point $x$
and some finite time $T_{0}$;\newline(2)The derivative of the solution becomes
infinitely large at some point $x$ and some finite time $T_{0}$.
\end{definition}

In this paper, we concern about blowup solutions for the $3$-dimensional
pressureless Navier--Stokes-Poisson equations with the density-dependent
viscosity, which may describe the phenomenon called the core collapsing in
evolution of gas star. And our aim is to construct a family of such blowup
solutions to it.

Historically in astrophysics, Goldreich and Weber \cite{GW} constructed the
analytical blowup solution (collapsing) of the $3$-dimensional Euler-Poisson
equation for $\gamma=4/3$ for the non-rotating gas spheres. After that, Makino
\cite{M1} obtained the rigorously mathematical proof of the existence of such
kind of blowup solutions. And in \cite{DXY}, we find the extension of the
above blowup solutions to the case . In \cite{Y}, the solutions with a from is
rewritten as

For $N\geq3$ and $\gamma=(2N-2)/N$,
\begin{equation}
\left\{
\begin{array}
[c]{c}%
\rho(t,r)=\left\{
\begin{array}
[c]{c}%
\dfrac{1}{a(t)^{N}}y(\frac{r}{a(t)})^{N/(N-2)},\text{ for }r<a(t)Z_{\mu};\\
0,\text{ for }a(t)Z_{\mu}\leq r.
\end{array}
\right.  \text{, }{\normalsize u(t,r)=}\dfrac{\dot{a}(t)}{a(t)}%
{\normalsize r,}\\
\ddot{a}(t){\normalsize =}-\dfrac{\lambda}{a(t)^{N-1}},\text{ }%
{\normalsize a(0)=a}_{0}>0{\normalsize ,}\text{ }\dot{a}(0){\normalsize =a}%
_{1},\\
\ddot{y}(z){\normalsize +}\dfrac{N-1}{z}\dot{y}(z){\normalsize +}\dfrac
{\alpha(N)}{(2N-2)K}{\normalsize y(z)}^{N/(N-2)}{\normalsize =\mu,}\text{
}y(0)=\alpha>0,\text{ }\dot{y}(0)=0,
\end{array}
\right.  \label{solution2}%
\end{equation}
where $\mu=[N(N-2)\lambda]/(2N-2)K$ and the finite $Z_{\mu}$ is the first zero
of $y(z)$;

For $N=2$ and $\gamma=1,$%
\begin{equation}
\left\{
\begin{array}
[c]{c}%
\rho(t,r)=\dfrac{1}{a(t)^{2}}e^{y(r/a(t))}\text{, }{\normalsize u(t,r)=}%
\dfrac{\dot{a}(t)}{a(t)}{\normalsize r;}\\
\ddot{a}(t){\normalsize =}-\dfrac{\lambda}{a(t)},\text{ }{\normalsize a(0)=a}%
_{0}>0{\normalsize ,}\text{ }\dot{a}(0){\normalsize =a}_{1};\\
\ddot{y}(x){\normalsize +}\dfrac{1}{x}\dot{y}(x){\normalsize +\dfrac
{\alpha(N)}{K}e}^{y(x)}{\normalsize =\mu,}\text{ }y(0)=\alpha,\text{ }\dot
{y}(0)=0,
\end{array}
\right.  \label{solution 3}%
\end{equation}
where $K>0$, $\mu=2\lambda/K$ with a sufficiently small $\lambda$ and $\alpha$
are constants.

For the construction of special analytical solutions to Navier-Stokes
equations, readers may refer Yuen's recent results in \cite{Y2}.

In this article, We extend the isothermal $R^{2}$ blowup solutions to the
Euler-Poisson equations to the pressureless Navier-Stokes-Poisson equations
with density-dependent viscosity in $R^{3}$with $\theta=1$ in spherical
symmetry:%
\begin{equation}
\left\{
\begin{array}
[c]{rl}%
\rho_{t}+u\rho_{r}+\rho u_{r}+{\normalsize \dfrac{2}{r}\rho u} &
{\normalsize =0,}\\
\rho\left(  u_{t}+uu_{r}\right)  & {\normalsize =-}\dfrac{4\pi\rho}{r^{2}}%
%TCIMACRO{\dint _{0}^{r}}%
%BeginExpansion
{\displaystyle\int_{0}^{r}}
%EndExpansion
\rho(t,s)s^{2}ds+[\kappa\rho]_{r}u_{r}+(\kappa\rho)(u_{rr}+\dfrac{2}{r}%
u_{r}-\dfrac{2}{r^{2}}u),
\end{array}
\right.  \label{gamma=1}%
\end{equation}
in the form of the following theorem.

\begin{theorem}
\label{thm:1}For the $3$-dimensional pressureless Navier--Stokes-Poisson
equations with $\theta=1$, in spherical symmetry, (\ref{gamma=1}), there
exists a family of solutions,%
\begin{equation}
\left\{
\begin{array}
[c]{c}%
\rho(t,r)=\dfrac{1}{(T-Ct)^{3}}e^{y(r/(T-Ct))}\text{, }{\normalsize u(t,r)=}%
\dfrac{-C}{T-Ct}{\normalsize r;}\\
\ddot{y}(x){\normalsize +}\dfrac{2}{x}\dot{y}(x){\normalsize +\dfrac{4\pi
}{\kappa C}e}^{y(x)}{\normalsize =0,}\text{ }y(0)=\alpha,\text{ }\dot{y}(0)=0,
\end{array}
\right.  \label{solution1}%
\end{equation}
where $T>0$, $\kappa>0$, $C>0$ and $\alpha$ are constants.\newline And the
solutions blowup in the finite time $T/C$.
\end{theorem}

\section{Separable Blowup Solutions}

In this section, before presenting the proof of Theorem \ref{thm:1}, we
prepare some lemmas. First, we obtain a general class of solutions for the
continuity equation of mass in radial symmetry (\ref{gamma=1})$_{1}$.

\begin{lemma}
\label{lem:generalsolutionformasseq}For the 3-dimensional conservation of mass
in spherical symmetry
\begin{equation}
\rho_{t}+u\rho_{r}+\rho u_{r}+\dfrac{2}{r}\rho u=0,
\label{massequationspherical}%
\end{equation}
there exist solutions,%
\begin{equation}
\rho(t,r)=\frac{1}{(T-Ct)^{3}}e^{f(r/(T-Ct))},\text{ }{\normalsize u(t,r)=}%
\frac{-C}{T-Ct}{\normalsize r,} \label{generalsolutionformassequation}%
\end{equation}
where $T$ and $C$ are positive constants and $f\in C^{1}$ is a non-negative
function .
\end{lemma}

\begin{proof}
We just plug (\ref{generalsolutionformassequation}) into
(\ref{massequationspherical}). Then
\begin{align*}
&  \rho_{t}+u\rho_{r}+\rho u_{r}+\dfrac{2}{r}\rho u\\
&  =\frac{(-3)(-C)e^{f(r/(T-Ct))}}{(T-Ct)^{4}}+\frac{e^{f(r/(T-Ct))}\dot
{f}(r/(T-Ct))}{(T-Ct)^{3}}\frac{r(-1)(-C)}{(T-Ct)^{2}}\\
&  +\frac{(-C)}{T-Ct}{\normalsize r}\frac{e^{f(r/(T-Ct))}\dot{f}%
(r/(T-Ct))}{(T-Ct)^{3}}\frac{1}{T-Ct}+\frac{e^{f(r/(T-t))}}{(T-Ct)^{3}}%
\frac{(-C)}{T-Ct}{\normalsize +}\dfrac{2}{r}\frac{e^{f(r/(T-Ct))}}{(T-Ct)^{3}%
}\frac{(-C)}{T-Ct}{\normalsize r}\\
&  =\frac{3Ce^{f(r/(T-Ct))}}{(T-Ct)^{4}}+\frac{Ce^{f(r/(T-Ct))}\dot
{f}(r/(T-Ct))}{(T-Ct)^{5}}r\\
&  -\frac{Ce^{f(r/(T-Ct))}\dot{f}(r/(T-Ct))}{(T-Ct)^{5}}r-\frac
{Ce^{f(r/(T-Ct))}}{(T-Ct)^{4}}-\frac{2Ce^{f(r/(T-Ct))}}{(T-Ct)^{4}}\\
&  =0.
\end{align*}
The proof is completed.
\end{proof}

Besides, we need the two lemmas for stating the property of the function
$y(z)$. The similar lemmas are already given in (Lemmas 9 and 10) Yuen's
article \cite{Y1}. For the completeness of the article, the proof of the below
lemmas are presented.

\begin{lemma}
\label{lemma2}There exists a sufficiently small $x_{0}>0$, such that the
equation%
\begin{equation}
\left\{
\begin{array}
[c]{c}%
\ddot{y}(x){\normalsize +}\dfrac{2}{x}\dot{y}(x){\normalsize +\sigma e}%
^{y(x)}{\normalsize =0,}\\
y(0)=\alpha,\text{ }\dot{y}(0)=0,
\end{array}
\right.  \label{SecondorderElliptic}%
\end{equation}
where $\sigma>0$ and $\alpha$ are constants, has a solution $y=y(x)\in
C^{2}[0,x_{0}]$.
\end{lemma}

\begin{proof}
The lemma can be proved by the fixed point theorem. Multiply
(\ref{SecondorderElliptic}) by $x$, to\ give%
\[
\frac{d}{dx}\left(  x^{2}\dot{y}(x)\right)  =-\sigma e^{y(x)}x^{2}.
\]
Notice $\dot{y}(0)=0$,\ we have%
\[
\dot{y}(x)=\frac{-1}{x^{2}}\int_{0}^{x}\sigma e^{y(s)}s^{2}ds.
\]
By using $y(0)=\alpha$, (\ref{SecondorderElliptic}) is reduced to%
\[
\dot{y}(x)=\frac{-1}{x^{2}}\int_{0}^{x}\sigma e^{y(s)}s^{2}ds,\text{
}y(0)=\alpha.
\]
Set%
\[
{\normalsize f(x,y(x))=}\frac{-1}{x^{2}}\int_{0}^{x}\sigma e^{y(s)}%
{\normalsize s}^{2}{\normalsize ds.}%
\]
then for any $x_{0}>0$, we get $f\in C^{1}[0,$ $x_{0}]$. and for any $y_{1,}$
$y_{2}\in C^{2}[0,$ $x_{0}]$, we have,%
\[
\left\vert f(x,y_{1}(x))-f(x,y_{2}(x))\right\vert =\frac{\sigma\left\vert
\int_{0}^{x}s^{2}(e^{y_{2}(s)}-e^{y_{1}(s)})ds\right\vert }{x^{2}}.
\]
As $e^{y}$ is a $C^{1}$ function of $y$, we can show that the function $e^{y}%
$, is Lipschitz-continuous. And we get,%
\begin{align*}
&  \left\vert f(x,y_{1}(x))-f(x,y_{2}(x)\right\vert \\
&  =\frac{O(1)\int_{0}^{x}s^{2}\left\vert \left(  y_{2}(s\right)
-y_{1}(s)\right\vert ds}{x^{2}}\\
&  \leq O(1)x_{0}\underset{0\leq x\leq x_{0}}{\sup}\left\vert y_{1}%
(s)-y_{2}(s)\right\vert ,
\end{align*}
where $\tau\in\lbrack0,$ $x]\subseteq\lbrack0,$ $x_{0}]$. Let%
\[
{\normalsize Ty(x)=\alpha+}\int_{0}^{x}{\normalsize f(s,y(s))ds.}%
\]
We have $Ty\in C[0,$ $x_{0}]$\ and%
\begin{align*}
&  \left\vert Ty_{1}(x)-Ty_{2}(x)\right\vert \\
&  =\left\vert \int_{0}^{x}f(s,y_{1}(s))ds-\int_{0}^{x}f(s,y_{2}%
(s))ds\right\vert \\
&  \leq O(1)x_{0}\underset{0\leq x\leq x_{0}}{\sup}\left\vert y(x)_{1}%
-y(x)_{2}\right\vert .
\end{align*}
By choosing $x_{0}>0$ to be a sufficiently small number, such that
$O(1)x_{0}<1$, this shows that the mapping $T:C[0,$ $X_{0}]\rightarrow C[0,$
$x_{0}]$, is a contraction with the sup-norm. By the fixed point theorem,
there exists a unique $y(x)\in C[0,$ $x_{0}],$\ such that $Ty(x)=y(x)$. The
proof is completed.
\end{proof}

\begin{lemma}
\label{lemma3}The ODE,%
\begin{equation}
\left\{
\begin{array}
[c]{c}%
\ddot{y}(x){\normalsize +}\dfrac{2}{x}\dot{y}(x){\normalsize +\sigma e}%
^{y(x)}{\normalsize =0,}\\
y(0)=\alpha,\text{ }\dot{y}(0)=0,
\end{array}
\right.  \label{Elliptic1}%
\end{equation}
where $\sigma>0$ and $\alpha$ are constants, has a solution in $[0,$
$+\infty)$ and $\underset{x\rightarrow+\infty}{\lim}y(x)=-\infty$.
\end{lemma}

\begin{proof}
By integrating (\ref{Elliptic1}), we have,%
\begin{equation}
\dot{y}(x)=-\frac{\sigma}{x^{2}}\int_{0}^{x}e^{y(s)}s^{2}ds\leq0.
\label{lemma3eq1}%
\end{equation}
Thus, for $0<x<x_{0}$, $y(x)$ has a uniform lower upper bound
\[
y(x)\leq y(0)=\alpha.
\]
As we obtained he local existence in Lemma \ref{lemma2}, there are two
possibilities:\newline(1)$y(x)$ only exists in some finite interval
$[0,x_{0}]$:\newline(1a)$\underset{x\rightarrow x_{0-}}{\lim}y(x)=-\infty
$;\newline(1b)$y(x)$ has an uniformly lower bound, i.e. $y(x)\geq\alpha_{0}$
for some constant $\alpha_{0}.$\newline(2)$y(x)$ exists in $[0,$ $+\infty
)$:\newline(2a)$\underset{x\rightarrow+\infty}{\lim}y(x)=-\infty$%
;\newline(2b)$y(x)$ has an uniformly lower bound, i.e. $y(x)\geq\alpha_{0}$
for some constant $\alpha_{0}$.\newline We claim that possibility (1) doesn't
exist. We need to reject (1b) first: If the statement (1b) is true,
(\ref{lemma3eq1}) becomes%
\begin{equation}
-\frac{\sigma xe^{\alpha_{0}}}{3}=-\frac{\sigma}{x^{2}}\int_{0}^{x}%
e^{\alpha_{0}}s^{2}ds\leq\dot{y}(x). \label{possible1}%
\end{equation}
Thus, $\dot{y}(x)$ is bounded in $[0,x_{0}]$. Therefore, we can use the fixed
point theorem again to obtain a large domain of existence, such that
$[0,x_{0}+\delta]$ for some positive number $\delta$. There is a
contradiction. Therefore, (1b) is rejected.\newline Next, we don't accept (1a)
because of the following reason: It is impossible that $\underset{x\rightarrow
x_{0-}}{\lim}y(x)=-\infty$, as from (\ref{possible1}), $\dot{y}(x)$ has a
lower bound in $[0,$ $x_{0}]$:%
\begin{equation}
-\frac{\sigma x_{0}e^{\alpha}}{3}\leq\dot{y}(x). \label{lemma3eq2}%
\end{equation}
Thus, (\ref{lemma3eq2}) becomes,
\begin{align*}
y(x_{0})  &  =y(0)+\int_{0}^{x_{0}}\dot{y}(x)dx\\
&  \geq\alpha-\int_{0}^{x_{0}}\frac{\sigma xe^{\alpha}}{3}dx\\
&  =\alpha-\frac{\sigma x_{0}^{2}e^{\alpha}}{6}.
\end{align*}
Since $y(x)$ is bounded below in $[0,$ $x_{0}]$, it contracts the statement
(1a), such that $\underset{x\rightarrow x_{0-}}{\lim}y(x)=-\infty$. So, we can
exclude the possibility (1).\newline We claim that the possibility (2b)
doesn't exist. It is because
\[
\dot{y}(x)=-\dfrac{\sigma}{x^{2}}\int_{0}^{x}e^{y(s)}s^{2}ds\leq-\frac{\sigma
}{x^{2}}\int_{0}^{x}e^{\alpha_{0}}s^{2}ds=-\frac{\sigma e^{a_{0}}x}{3}.
\]
Then, we have,%
\begin{equation}
y(x)\leq\alpha-\frac{\sigma e^{a_{0}}}{6}x^{2}. \label{lemma3eq3}%
\end{equation}
By letting $x\rightarrow\infty$, (\ref{lemma3eq3}) turns out to be,
\[
y(x)=-\infty.
\]
Since a contradiction is established, we exclude the possibility (2b). Thus,
the equation (\ref{Elliptic1}) exists in $[0,$ $+\infty)$ and $\underset
{x\rightarrow+\infty}{\lim}y(x)=-\infty$. This completes the proof.
\end{proof}

Up to this stage, we are already to give the proof of Theorem \ref{thm:1}.

\begin{proof}
[Proof of Theorem 2]From Lemma \ref{lem:generalsolutionformasseq}, we easily
get that (\ref{solution1}) satisfy (\ref{gamma=1})$_{1}$. For the momentum
equation (\ref{gamma=1})$_{2}$, we get,%
\begin{align}
&  \rho(u_{t}+uu_{r})+\frac{4\pi\rho}{r^{2}}%
%TCIMACRO{\dint \limits_{0}^{r}}%
%BeginExpansion
{\displaystyle\int\limits_{0}^{r}}
%EndExpansion
\rho(t,s)s^{2}ds-[\mu(\rho)]_{r}u_{r}-\mu(\rho)(u_{rr}+\dfrac{2}{r}%
u_{r}-\dfrac{2}{r^{2}}u)\\
&  =\rho\left[  \frac{(-C)(-1)(-C)}{(T-Ct)^{2}}r+\frac{(-C)}{T-Ct}r\cdot
\frac{(-C)}{T-Ct}\right]  +\frac{4\pi\rho}{r^{2}}%
%TCIMACRO{\dint \limits_{0}^{r}}%
%BeginExpansion
{\displaystyle\int\limits_{0}^{r}}
%EndExpansion
\frac{e^{y(s/(T-Ct))}}{(T-Ct)^{3}}s^{2}ds\nonumber\\
&  -(\kappa\rho)_{r}\frac{(-C)}{T-Ct}-\mu(\rho)\left(  0+\dfrac{2}{r}%
\frac{(-1)}{T-Ct}-\frac{2}{r^{2}}\frac{(-1)}{T-Ct}r\right) \nonumber\\
&  =\frac{\kappa Ce^{y(r/(T-Ct))}}{(T-Ct)^{3}}\cdot\dot{y}(\frac{r}%
{T-Ct})\cdot\frac{1}{T-Ct}\frac{C}{T-Ct}+\frac{4\pi\rho}{r^{2}}%
%TCIMACRO{\dint \limits_{0}^{r}}%
%BeginExpansion
{\displaystyle\int\limits_{0}^{r}}
%EndExpansion
\frac{e^{y(s/(T-Ct))}}{(T-Ct)^{3}}s^{2}ds\nonumber\\
&  =\frac{\kappa C\rho}{(T-Ct)^{2}}\dot{y}(\frac{r}{T-Ct})+\frac{4\pi\rho
}{r^{2}}%
%TCIMACRO{\dint \limits_{0}^{r}}%
%BeginExpansion
{\displaystyle\int\limits_{0}^{r}}
%EndExpansion
\frac{e^{y(s/(T-Ct))}}{(T-Ct)^{3}}s^{2}ds\nonumber\\
&  =\frac{\rho}{(T-Ct)^{2}}\left[  \kappa C\dot{y}(\frac{r}{T-Ct})+\frac{4\pi
}{r^{2}(T-Ct)}%
%TCIMACRO{\dint \limits_{0}^{r}}%
%BeginExpansion
{\displaystyle\int\limits_{0}^{r}}
%EndExpansion
e^{y\left(  s/(T-Ct)\right)  }s^{2}ds\right]  .
\end{align}
By letting $\omega=s/(T-Ct)$, it follows:%
\begin{align}
&  =\frac{\rho}{(T-Ct)^{2}}\left[  \kappa C\dot{y}(\frac{r}{T-Ct})+\frac{4\pi
}{(\frac{r}{T-Ct})^{2}}%
%TCIMACRO{\dint \limits_{0}^{r/(T-Ct)}}%
%BeginExpansion
{\displaystyle\int\limits_{0}^{r/(T-Ct)}}
%EndExpansion
e^{y\left(  s\right)  }s^{2}ds\right] \\
&  =\frac{\rho}{(T-Ct)^{2}}Q\left(  \frac{r}{T-Ct}\right)  .\nonumber
\end{align}
And denote $z=r/(T-ct)$,%
\[
Q(\frac{r}{T-Ct})={\normalsize Q(z)=}\kappa C\dot{y}(z){\normalsize +}%
\dfrac{4\pi}{z^{2}}%
%TCIMACRO{\dint \limits_{0}^{z}}%
%BeginExpansion
{\displaystyle\int\limits_{0}^{z}}
%EndExpansion
e^{y\left(  \omega\right)  }\omega^{2}d\omega{\normalsize .}%
\]
Differentiate $Q(z)$\ with respect to $z$,%
\begin{align*}
\dot{Q}(z)  &  =\kappa C\ddot{y}(z){\normalsize +4\pi e}^{y(z)}+\frac
{(-2)\cdot4\pi}{z^{3}}%
%TCIMACRO{\dint \limits_{0}^{x}}%
%BeginExpansion
{\displaystyle\int\limits_{0}^{x}}
%EndExpansion
e^{y(\omega)}{\normalsize \omega}^{2}{\normalsize ds}\\
&  =-\frac{2\kappa C}{z}\dot{y}(z)+\frac{(-2)4\pi}{z^{3}}%
%TCIMACRO{\dint \limits_{0}^{z}}%
%BeginExpansion
{\displaystyle\int\limits_{0}^{z}}
%EndExpansion
e^{y(\omega)}{\normalsize \omega}^{2}{\normalsize d\omega}\\
&  =-\frac{2}{z}Q(z),
\end{align*}
where the above result is due to the fact that we choose the following
ordinary differential equation:%
\[
\left\{
\begin{array}
[c]{c}%
\ddot{y}(z){\normalsize +}\dfrac{2}{z}\dot{y}(z){\normalsize +\dfrac{4\pi
}{\kappa C}e}^{y(z)}{\normalsize =0}\\
{\normalsize y(0)=\alpha,}\text{ }\dot{y}(0){\normalsize =0.}%
\end{array}
\right.
\]
With $Q(0)=0$, this implies that $Q(z)=0$. Thus, the momentum equation
(\ref{gamma=1})$_{2}$ is satisfied.\newline With the Lemma \ref{lemma3} about
$y(z)$, we are able to show that the family of the solutions blows up in
finite time $T/C$. This completes the proof.
\end{proof}

The statement about the blowup rate will be immediately followed:

\begin{corollary}
The blowup rate of the solution (\ref{solution1}) is
\[
\underset{t\rightarrow T/C}{\lim}\rho(t,0)(T-Ct)^{3}\geq O(1).
\]

\end{corollary}

Given that the sign of the constant $C$ in (\ref{solution 3}) is changed to be
negative, the below corollary is clearly shown.

\begin{corollary}
For the $3$-dimensional pressureless Navier--Stokes-Poisson equations with
$\theta=1$ in spherical symmetry, (\ref{gamma=1}), there exists a family of
solutions,%
\begin{equation}
\left\{
\begin{array}
[c]{c}%
\rho(t,r)=\dfrac{1}{(T-Ct)^{3}}e^{y(r/(T-Ct))}\text{, }{\normalsize u(t,r)=}%
\dfrac{-C}{T-Ct}{\normalsize r;}\\
\ddot{y}(x){\normalsize +}\dfrac{2}{x}\dot{y}(x){\normalsize +\dfrac{4\pi
}{\kappa C}e}^{y(x)}{\normalsize =0,}\text{ }y(0)=\alpha,\text{ }\dot{y}(0)=0,
\end{array}
\right.
\end{equation}
where $T>0$, $\kappa>0$, $C<0$ and $\alpha$ are constants.
\end{corollary}

At last, we give some remarks to discuss the constructed blowup solutions.

\begin{remark}
Our solutions are the modification versions of blowup solutions
(\ref{solution 3}) of Yuen in the $2$-dimensional Euler-Poisson equations.
However, we emphasize that the process is not trivial, as the blowup solutions
cannot be generated to the other dimensional case $(N\neq2)$. And the blowup
solutions (\ref{solution2}) of Goldreich, Weber and Makino also cannot be
generated to the lower dimensions'($1$-dimensional or $2$-dimensional) case.
\end{remark}

\begin{remark}
It is well-known that it is an open problem to construct the analytical blowup
solutions in the $3$-dimensional case to understand the blowup phenomena
(black hole) in the Euler-Poisson equations in astrophysics. But for the
non-trivial blowup solutions, there only exists the blowup solutions
(\ref{solution2}). Meanwhile, due to the solution structure of $u$ in
(\ref{solution2})\ and (\ref{solution 3}),%
\[
u=\frac{\dot{a}(t)}{a(t)}r,
\]
where $a(t)\neq0$ is a function of time,\newline that automatically satisfies
that the density-independent viscosity function $(\theta=0)$:%
\[
vis(u)=k\Delta u=\kappa(u_{rr}+\frac{N-1}{r}u_{r}-\frac{N-1}{r^{2}}u)=0.
\]
In this article, we shift the scope of consideration from the Euler-Poisson
equations in to the Navier-Stokes-Poisson equations with density-dependent
viscosity. We can extend the isothermal $R^{2}$ blowup solutions to the
Euler-Poisson equations to pressureless Navier-Stokes-Poisson equations with
density-dependent viscosity in $R^{3}$. Here blowup solutions can answer the
above problem in the modified $3$-dimensional fluid dynamic system. We notice
that the extension of such blowup solutions are not suitable for solutions
(\ref{solution2}) from $R^{N}$ to $R^{N+1}$ ($N\geq3$). In the future, the
blowup solutions in the Euler-Poisson equations will still be sought.
\end{remark}

\begin{remark}
Besides, if we consider the $3$-dimensional drift-diffusion model in
semiconductors in spherical symmetry,%
\[
\left\{
\begin{array}
[c]{rl}%
\rho_{t}+u\rho_{r}+\rho u_{r}+{\normalsize \dfrac{2}{r}\rho u} &
{\normalsize =0,}\\
\rho\left(  u_{t}+uu_{r}\right)  & {\normalsize =+}\dfrac{4\pi\rho}{r^{2}}%
\int_{0}^{r}\rho(t,s)s^{2}ds+[\kappa\rho]_{r}u_{r}+(\kappa\rho)(u_{rr}%
+\dfrac{2}{r}u_{r}-\dfrac{2}{r^{2}}u),
\end{array}
\right.
\]
the special solutions with infinite mass may be obtained as follows:%
\begin{equation}
\left\{
\begin{array}
[c]{c}%
\rho(t,r)=\dfrac{1}{(T-Ct)^{3}}e^{y(r/(T-Ct))}\text{, }{\normalsize u(t,r)=}%
\dfrac{-C}{T-Ct}{\normalsize r;}\\
\ddot{y}(x){\normalsize +}\dfrac{2}{x}\dot{y}(x)-{\normalsize \dfrac{4\pi
}{\kappa C}e}^{y(x)}{\normalsize =0,}\text{ }y(0)=\alpha,\text{ }\dot{y}(0)=0
\end{array}
\right.
\end{equation}

\end{remark}

\end{document}